\documentclass[a4paper]{article} 
\usepackage{cite}
\usepackage{charter}
\usepackage{fullpage}
\usepackage[toc,page]{appendix}
\usepackage{graphics} 
\usepackage{mathptmx} 
\usepackage{times}    
\usepackage{amsmath}  
\usepackage{amssymb}  
\usepackage{helvet}   
\usepackage{epsfig} 
\usepackage{scalefnt}
\usepackage[english]{babel}
\usepackage{epsfig,graphicx}
\usepackage{tabularx}
\usepackage{helvet}
\usepackage{scalefnt} 
\usepackage{verbatim}
\usepackage{pslatex}
\usepackage{xcolor}
\usepackage[boxruled, vlined]{algorithm2e}
\usepackage[all, knot]{xy}
\xyoption{arc}

\newcommand{\equals}{\stackrel{\mathrm{def}}{=}}

\newtheorem{theorem}{Theorem} 
\newtheorem{lemma}{Lemma}

\newenvironment{proof}[1][Proof]{\begin{trivlist}
\item[\hskip \labelsep {\bfseries #1}]}{\end{trivlist}}

\newcommand{\CPU}{\operatorname{CPU}}
\newcommand{\FFD}{\operatorname{FFD}}

\newcommand{\EDF}{\operatorname{EDF}}
\newcommand{\DBF}{\operatorname{DBF}}
\newcommand{\WCET}{\operatorname{WCET}}
\newcommand{\Intel}{\operatorname{Intel}}
\newcommand{\AMD}{\operatorname{AMD}}
\newcommand{\Sun}{\operatorname{Sun}}
\newcommand{\Schedjobs}{\operatorname{Schedjobs}}
\newcommand{\Sched}{\operatorname{Sched}}
\newcommand{\True}{\operatorname{True}}
\newcommand{\False}{\operatorname{False}}
\newcommand{\Free}{\operatorname{Free}}
\newcommand{\Load}{\operatorname{Load}}
\newcommand{\RemainingJobs}{\operatorname{RemJobs}}

\newcommand{\fin}{\hfill{\small $\blacksquare$}}     
\usepackage{color}

\begin{document}

\title{Semi-Partitioned Hard Real-Time Scheduling with Restricted Migrations upon Identical Multiprocessor Platforms}


\author{
	Fran\c cois Dorin$^1$\thanks{PhD candidate of the LISI - ENSMA, France.} \\
	dorinfr@ensma.fr 
	\and
	Patrick Meumeu Yomsi$^2$\thanks{Postdoctoral researcher of the FNRS, Belgium.}\\
	patrick.meumeu.yomsi@ulb.ac.be 
	\and
	Jo\"{e}l Goossens$^2$\\
	joel.goossens@ulb.ac.be 
	\and 
	Pascal Richard$^1$ \\
	richardp@ensma.fr 
}

\maketitle

\footnotetext[1]{LISI - ENSMA - Universit\'e de Poitiers -- 1 Rue C. Ader, \\ \phantom{xxxx}B.P. 40109 -- 86961 Chasseneuil du Poitou - France.}
\footnotetext[2]{Universit\'{e} Libre de Bruxelles (ULB) -- 50 Avenue F. D. Roosevelt, \\ \phantom{xxx} C.P. 212 -- 1050 Brussels - Belgium.}
\addtocounter{footnote}{2}
                                                                                                         

\begin{abstract}
Algorithms based on {\em semi-partitioned} scheduling have been proposed as a viable alternative between the two extreme ones based on {\em global} and {\em partitioned} scheduling. In particular, allowing migration to occur only for few tasks which cannot be assigned to any individual processor, while most tasks are assigned to specific processors, considerably reduces the runtime overhead compared to {\em global} scheduling on the one hand, and improve both the schedulability and the system utilization factor compared to {\em partitioned} scheduling on the other hand.
  
In this paper, we address the preemptive scheduling problem of hard real-time systems composed of sporadic constrained-deadline tasks upon {\em identical} multiprocessor platforms. We propose a new algorithm and a scheduling paradigm based on the concept of semi-partitioned scheduling with restricted migrations in which jobs are not allowed to migrate, but two subsequent jobs of a task can be assigned to different processors by following a {\em periodic strategy}. 
\end{abstract}

\vspace{-0.3cm}
\section{Introduction}\label{Introduction}

In this work, we consider the preemptive scheduling of hard real-time sporadic tasks upon {\em identical} multiprocessors. Even though the {\em Earliest Deadline First} ($\EDF$) algorithm \cite{Liu73} turned out to be no longer optimal in terms of schedulability upon multiprocessor platforms \cite{dhall78}, many alternative algorithms based on this scheduling policy have been developed due to its optimality upon uniprocessor platforms \cite{Liu73}. Again, the primary focus for designers is at improving the worst-case system utilization factor with guaranteeing all tasks to meet deadlines. Unfortunately, most of the algorithms developed in the literature suffer from a trade-off between the theoretical schedulability and the practical overhead of the system at run-time: achieving high system utilization factor leads to complex computations. Up to now, solutions are still widely discussed.

In recent years, {\em multicore architectures}, which include several processors upon a single chip, have been the preferred platform for many embedded applications. This is because they have been considered as a solution to the ``thermal roadblock'' imposed by single-core designs. Most chip makers such as $\Intel$, $\AMD$ and $\Sun$ have released dual-core chips, and a few designs with more than two cores have been released as well. However given such a platform, the question of scheduling hard real-time systems becomes more complicated, and thus, has received considerable attention \cite{Goossens1, Baker07}. For such systems, most results have been derived under either {\em global} or {\em partitioned} scheduling techniques. While either approach might be viable, each has serious drawbacks on this platform, and neither will likely utilize the system very well. 

In {\em global} scheduling \cite{Bertogna09}, all tasks are stored in a single priority-ordered queue. The global scheduler selects for execution the highest priority tasks from this queue. Unfortunately, algorithms based on this technique may lead to runtime overheads that are prohibitive. This is due to the fact that tasks are allowed to migrate from one processor ($\CPU$) to another in order to complete their executions and each migration cost may not be negligible. Moreover, Dhall et al. showed that global $\EDF$ may cause a deadline to be missed if the total utilization factor of a task set is slightly greater than $1$ \cite{dhall78}. 

In {\em partitioned} scheduling \cite{SanFis06}, tasks are first assigned statically to processors. Once the task assignment is done, each processor uses independently its local scheduler. Unfortunately, algorithms based on this technique may lead to task systems that are schedulable if and only if some tasks are {\em not} partitioned. Moreover, Lopez et al. showed that the total utilization factor of a schedulable system using this technique is at most $50\%$.

These two scheduling techniques are {\em incomparable} --- at least for priority driven schedules ---, that is, there are systems which are schedulable with partitioning and not by global and reversely. It thus follows that the effectiveness of a scheduling algorithm depends not only on its runtime overhead, but also its ability to schedule feasible task systems. 

Recent work \cite{Kato1, Andersson2, Kato2} came up with a novel and promising technique called {\em semi-partitioned scheduling} with the main objectives of reducing the runtime overhead and improving both the schedulability and the system utilization factor. In semi-partitioned scheduling techniques, most tasks, called {\em non-migrating tasks}, are fixed to specific processors in order to reduce runtime overhead, while few tasks, called {\em migrating tasks}, migrate across processors in order to improve both the schedulability and the system utilization factor.
However, the migration costs depend on the instant each migrating task is migrated during execution: migrating a task during the execution of one of its jobs is more time consuming than migrating it at the instant of its activation. As such, we can distinguish between two levels of migration: {\em job migration} where a job is allowed to execute on different processors \cite{Andersson2} and {\em task migration} where task migration is allowed, but job migration is {\em forbidden}. Between the two techniques, the task migration is the one which minimizes the migration costs. To the best of our knowledge, only one solution using the latter semi-partitioned scheduling technique was established but it is only applicable to soft real-time systems \cite{Anderson0}.

In this paper, we study the scheduling of hard real-time systems composed of sporadic constrained-deadline tasks upon {\em identical} multiprocessor platforms. For such platforms, all the processors have the same computing capacities. 

\paragraph{Related work.} 
On the road to solve the above mentioned problem, sound results using semi-partitioned scheduling techniques have been obtained by S. Kato in terms of both schedulability and system utilization factor. However, these results are all based on ``{\em job-splitting}'' strategies \cite{Kato2}, and consequently, may still lead to a prohibitive runtime overheads for the system. Indeed, the main idea of S. Kato consists in using a ``{\em job-splitting}'' (i.e., job migration is allowed) strategy based on a specific algorithm for tasks which cannot be scheduled by following the {\em First Fit Decreasing} ($\FFD$) algorithm \cite{Kato2009}. Figure~\ref{fig:shinpei} illustrates that the entire portion of job $\tau_k$ is not partitioned but it is split into $\tau_{k, 1}$, $\tau_{k, 2}$ and $\tau_{k, 3}$ which are partitioned upon $\CPU$s $\pi_1$, $\pi_2$ and $\pi_3$, respectively. The amount of each share is such a value that fills the assigning processor to capacity without timing violations. 

\begin{figure}[h]
\begin{center}
	\includegraphics[width=\linewidth]{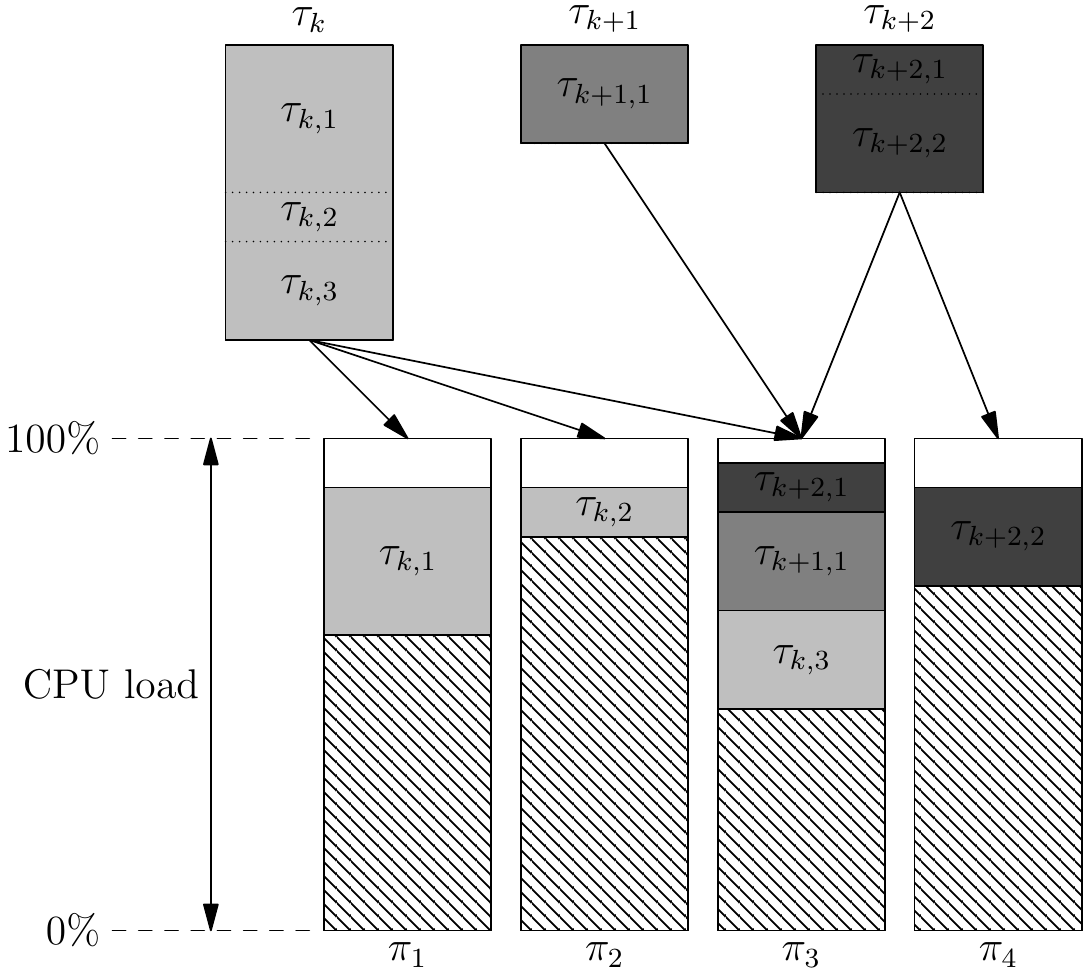}
	\caption{S. Kato's algorithm \cite{Kato2009}.}
	\label{fig:shinpei}
\end{center}
\end{figure}

Due to such a ``{\em job-splitting}'' strategy, it is necessary to have a mechanism which ensure that each job cannot be executed in parallel upon different processors. Such a mechanism has been provided by the introduction of {\em execution windows}\footnote{An execution window for a job is a time span during it is allowed to execute.}. Using this mechanism, each share can execute only during its execution window. Hence, multiple executions of the same job upon several processors at the same time are prohibited by guaranteeing that the execution windows do not overlap \cite{Kato2009}.

\paragraph{This research.} In this paper, we propose a new algorithm and a scheduling paradigm based on the concept of semi-partitioned scheduling with {\em restricted migrations}: jobs are not allowed to migrate, but two subsequent jobs of a task can be assigned to different processors by following a {\em periodic strategy}. Our intuitive idea is to limit the runtime overhead as it may still be prohibitive for the system when using a {\em job-splitting} strategy. The algorithm starts with a classical step where tasks are partitioned among processors using the $\FFD$ algorithm. Then, for the remaining tasks (i.e., those whose all the jobs cannot be assigned to one processor by following $\FFD$ without exceeding its capacity), a semi-partitioning scheduling technique with restricted migrations is used by following a periodic strategy. 

\paragraph{Paper organization.} The remainder of this paper is structured as follows. Section~\ref{sec:system_model} presents the task model and the platform that are used throughout the paper. Section~\ref{sec:algo_principle} provides the principles of the proposed algorithm. Section~\ref{sec:schedulability analysis}  elaborates the conditions under which the system is schedulable. Section~\ref{Experimental results} presents experimental results. Finally, Section~\ref{Conclusion and future work} concludes the paper and proposes future work.

\section{System model}\label{sec:system_model} 

We consider the preemptive scheduling of a hard real-time system $\tau \equals \{\tau_1, \tau_2, \ldots, \tau_{n}\}$ comprised of $n$ tasks upon $m$ {\em identical} processors. In this case, all the processors have the same computing capacities. The $k^{th}$ processor is denoted $\pi_k$. Each task $\tau_i$ is a sporadic constrained-deadline task characterized by three parameters $(C_i, D_i, T_i)$ where $C_i$ is the Worst Case Execution Time (WCET), $D_i \le T_i$ is the relative deadline and $T_i$ is the minimum inter-arrival time between two consecutive releases of $\tau_i$. These parameters are given with the interpretation that task $\tau_i$ generates an infinite number of successive jobs $\tau_{i,j}$ with execution requirement of at most $C_i$ each, arriving at time $a_{i,j}$ such that $a_{i,j+1} - a_{i,j} \ge T_i$ and that must completes within $[a_{i,j}, d_{i,j})$ where $d_{i,j} \equals a_{i,j} + D_i$. 

The {\em utilization factor} of task $\tau_i$ is defined as $u_i \equals C_i / T_i$ and the {\em total utilization factor} of task system $\tau$ is defined as $U_{\operatorname{sum}}(\tau) \equals \sum_{i=1}^{n} u_i$. A task system is said to {\em fully utilize} the available processing capacity if its total utilization factor equals the number of processors ($m$). The {\em maximum} utilization factor of any task in $\tau$ is denoted $u_{\operatorname{max}}(\tau)$. A task system is {\em preemptive} if the execution of its job may be interrupted and resumed later. In this paper, we consider preemptive scheduling policies and we place no constraints on total utilization factor.

We consider that tasks are scheduled by following the {\em Earliest Deadline First} ($\EDF$) scheduler with {\em restricted migrations}. That is, on the one hand, the shorter the deadline of a job the higher its priority. On the other hand, tasks can migrate from one processor to another but when a job has started upon a given processor then it must complete  without any migration. Such a strategy is a compromise between task partitioning and global scheduling strategies. 

We assume that all the tasks are independent, that is, there is no communication, no precedence constraint and no shared resource (except for the processors) between tasks. We assume that the jobs of a task cannot be executed in parallel, that is, for any $i$ and $j$, $\tau_{i, j}$ cannot be executed in parallel on more than one processor. Moreover, this research assumes that the costs of preempting and migrating tasks are included in the $\WCET$, which makes sense since we limit migrations at task arrival instants and preemptions at task arrival or completion instants ($\EDF$ is the local scheduler).

\section{Principle of the algorithm}\label{sec:algo_principle} 

In this section, we provide to the reader the main steps of our algorithm. Its design is based on the concept of semi-partitioned scheduling with {\em restricted migrations} and consists of two phases: \\
($1$) {\em The assigning phase}: here, each non-migrating task is assigned to a specific $\CPU$ by following the $\FFD$ algorithm and the jobs of each migrating tasks are assigned to $\CPU$s by following a {\em periodic strategy}. \\ 
($2$) {\em The scheduling phase}: here, each migrating task is modeled upon each $\CPU$ as a \emph{multiframe} task and thus the schedulability analysis focuses on each $\CPU$ individually, using the rich and extensive results from the uniprocessor scheduling theory.

\subsection{The assigning phase}\label{subsec:task assigning phase} 

In this phase, each task $\tau_k$ is assigned to a particular processor $\pi_j$ by following the $\FFD$ algorithm, as long as the task does not cause the system not to be schedulable upon $\pi_j$. That is, $\Load(\tau^{\pi_j}) \equals \sup_{(t \ge 0)} \frac{\DBF(\tau^{\pi_j}, t)}{t} \le 1$, where $\tau^{\pi_j}$ denotes the subset of tasks assigned to $\CPU$ $\pi_j$ and $\DBF$ is the classical Demand Bound Function \cite{Baruah1999}. Such a task is classified into a {\em non-migrating task}. Now, if $\Load(\tau^{\pi_j}) > 1$, that is, all the jobs of a task cannot be assigned to the same processor, then the task is classified into a {\em migrating task}. The jobs are assigned to  several processors by following a {\em periodic strategy} and are executed upon these by using a semi-partitioned scheduling with restricted migrations. Details on this strategy, obviously chosen for sake of simplicity during the implementations, are provided in the next section. 

\subparagraph{Example.} In order to illustrate this strategy for a system with a single {\em migrating task} task $\tau_i = (C_i, D_i, T_i)$, we assume that $\tau_i$ has been assigned to a set of $\CPU$s containing at least $\pi_1$ and $\pi_2$ with the periodic sequence $\sigma \equals (\pi_1, \pi_2, \pi_1)$. This means, the first job of $\tau_i$ will be assigned to $\pi_1$, the second one to $\pi_2$, the third one to processor $\pi_1$ and from the fourth job of $\tau_i$, this very same process will repeat {\em cyclically} upon processors $\pi_1$ and $\pi_2$. Figure~\ref{fig:jaa} depicts this job assignment. 

\begin{figure}[h]
\begin{center}
	\includegraphics[scale=0.6]{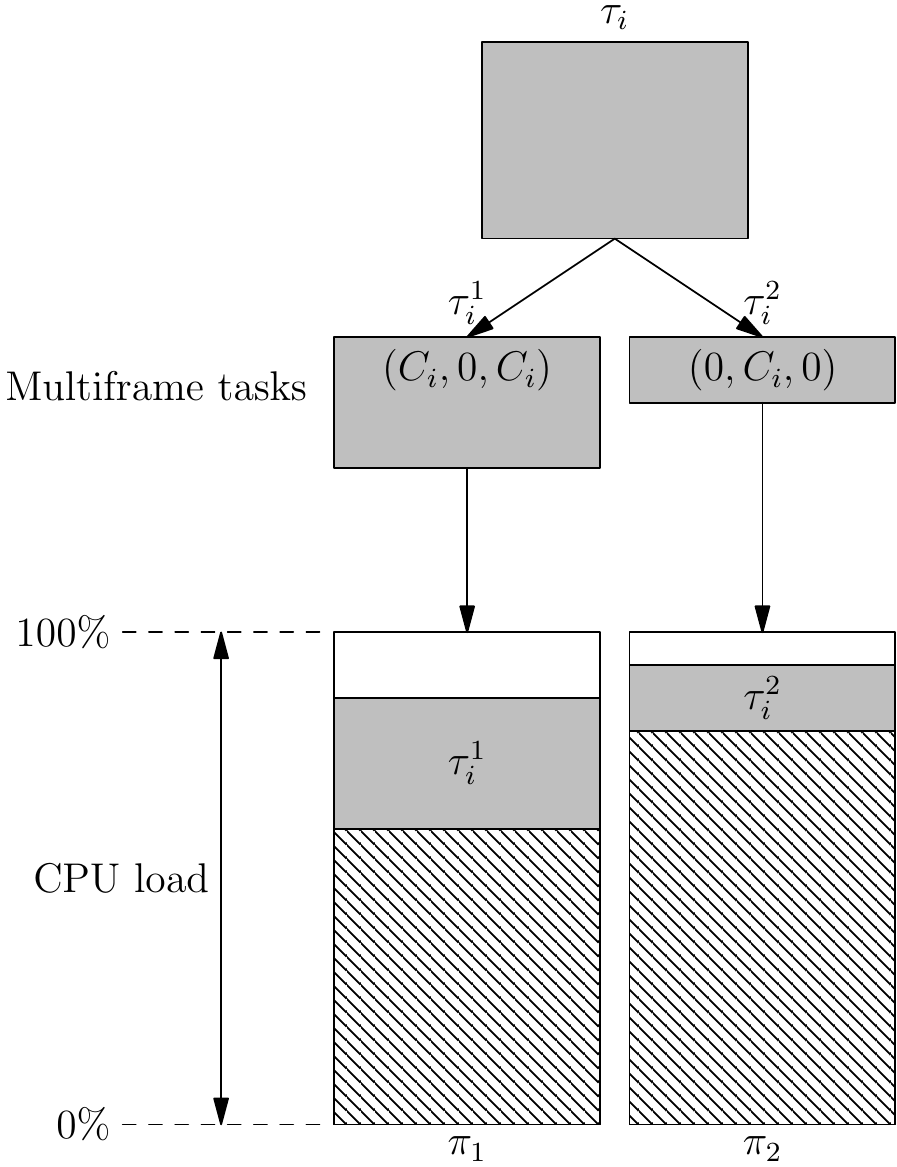}
	\caption{Periodic job assignment.}
	\label{fig:jaa}
\end{center}
\end{figure}

From the schedulability point of view, task $\tau_i$ will be duplicated upon $\CPU$s $\pi_1$ and $\pi_2$ and according to the scheduling phase, it will be seen as two multiframe tasks \cite{Mok3}. These multiframe tasks will be denoted by $\tau^1_{i}$ and $\tau^2_{i}$ and will have the following parameters: $\tau^1_{i} \equals ((C_i,0,C_i), D_i, T_i)$ and $\tau^2_{i} \equals ((0,C_i,0), D_i, T_i)$. As such, we will be able to perform the schedulability analysis upon each $\CPU$ individually, using uniprocessor approaches, already well understood.

\subsection{Periodic job assignment strategy}

Since each task consists of an infinite number of jobs, a potentially infinite number of periodic job assignment strategies can be defined for assigning migrating tasks to the $\CPU$s. In this section, we propose two algorithms to achieve this: The {\em most regular-job pattern algorithm} and the {\em alternative-job pattern algorithm}.

In this paper, since each migrating task is modeled as a multiframe task upon each $\CPU$ it has at least one job assigned to, we assume that the number of frames obtained from each {\em migrating task} is equal to a non-negative integer $K$ which is known beforehand for sake of readability. This assumption will be relaxed in future work. 

Letting $A[1\cdots n, 1\cdots m]$ be an matrix of integers where the first index refers to tasks and the second index to processors, we set the value $A[i,j] \equals x$ (with $1 \leq i \leq n$ and $1 \leq j \leq m$) to indicate that $x$ jobs among $K$ consecutive jobs of task $\tau_i$ will be executed upon processor $\pi_j$. In Section~\ref{subsec:scheduling algorithm}, we provide the reader with details on two algorithms used to initialize matrix $A[1\cdots n, 1\cdots m]$, while guaranteeing for every task $\tau_i$ that $\sum_{j=1, \ldots, m} A[i,j] = K$. 

Before going any further in this paper, it is worth noticing that if $K=1$, then task migrations are forbidden. Thus, our model extends the classical partitioning scheduling model. If the $A[i, j]$ are known beforehand, then the {\em most regular-job pattern algorithm} is considered, otherwise we consider the {\em alternative-job pattern algorithm}.

\subsection*{Most regular-job pattern algorithm}
A {\em uniform assignment} of jobs of each migrating task $\tau_i$ among the subset of $\CPU$s $\mathcal{S}_i$ upon which $\tau_i$ will execute at least one job seems to be a good idea at first glance (but we have no theoretical result proving that). For this reason we introduce the principle laying behind such a strategy \cite{Karp67}. In $\mathcal{S}_i$, we assume that the $\CPU$s are ranged in a non-decreasing index-order. If $\mathcal{S}_i = \emptyset$, then the system is clearly not schedulable. 

The job assignment is performed according to the following two steps for task $\tau_i$.

\begin{itemize}
\item \emph{Step 1.} The matrix $A$ is computed thanks to Algorithm~\ref{alg:split} (see Section~\ref{subsec:scheduling algorithm} for details) such that $A[i,j]$ jobs among $K$ consecutive jobs of task $\tau_i$ will be executed upon processor $\pi_j$ and such that 
$\sum_{j=1,\ldots,m} A[i, j] = K$.

\item \emph{Step 2.} The assignment sequence $\sigma$ of task $\tau_i$ is defined through $K$ sub-sequences 

\begin{equation}
 \label{eq:sigma_seq_init}
\sigma \equals (\sigma_0, \sigma_1, \ldots, \sigma_\ell, \ldots, \sigma_{K - 1})
\end{equation}
where the ${(\ell+1)}^{th}$ sub-sequence $\sigma_\ell$ (with $\ell = 0, \ldots, K-1$) is given in turn by the following $m$-tuple: \begin{equation}
 \label{eq:sigma_seq}
\sigma_\ell \equals (\sigma_\ell^1, \sigma_\ell^2, \ldots, \sigma_\ell^{m})
\end{equation}

To define sub-sequence $\sigma_\ell$ using the uniform assignment pattern, there is at most one job per $\CPU$ each time w.r.t.\@ Equation~\ref{eq:sigma}. The $\ell^{\operatorname{th}}$ job of $\tau_i$ will be assigned to $\pi_j$ if and only if:
\begin{equation} \label{eq:sigma}
 \sigma_\ell^j \equals \left\lceil \frac{\ell+1}{K} \cdot A[i,j]\right\rceil - \left\lceil \frac{\ell}{K} \cdot A[i, j]\right\rceil = 1
\end{equation}
The goal is to assign $A[i, j]$ jobs among $K$ consecutive jobs of task $\tau_i$ to $\CPU$ $\pi_j$ by following a step-case function. For the $\ell^{\operatorname{th}}$ job, $\sigma_\ell^j$ yields $1$ when the job is assigned to $\pi_j$ and $0$ otherwise. Figure~\ref{fig:aij} illustrates Equation~\ref{eq:sigma} for job assignment of $\tau_i$ to $\CPU$ $\pi_1$ when $K~=~11$ and $A[i, j] = 4$. 
\begin{figure}[h]
	\includegraphics[width=\linewidth]{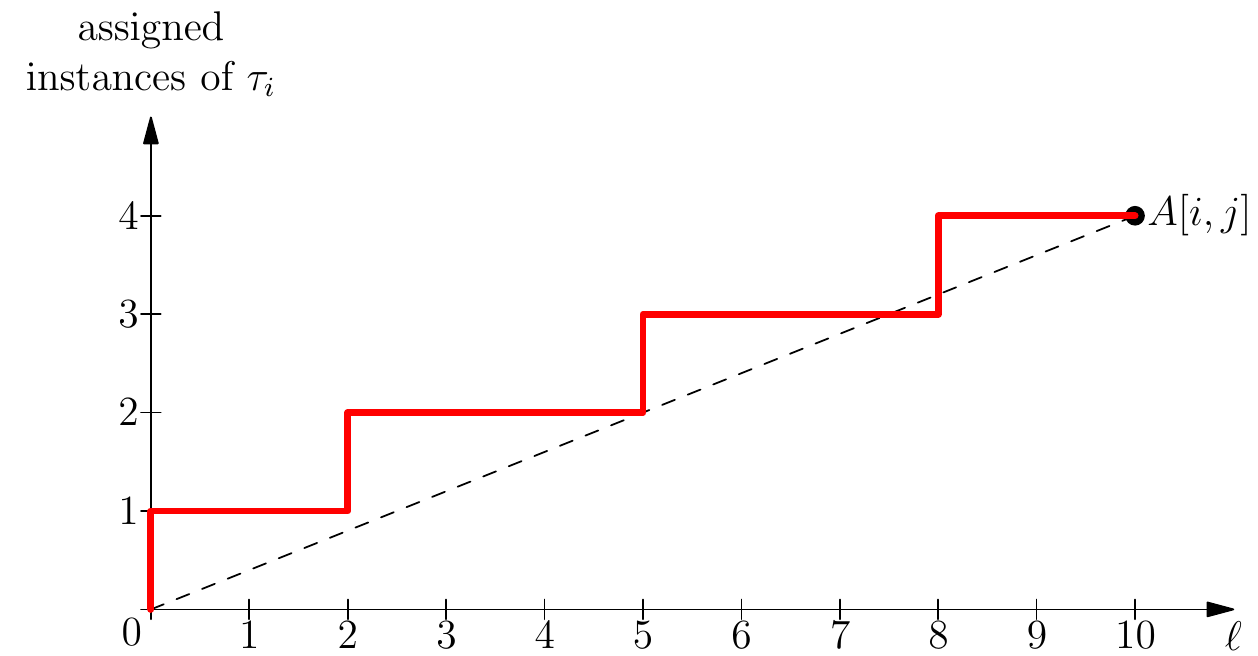}
	\caption{Job assignment sequence to $\pi_1$.}
	\label{fig:aij}
\end{figure}
\end{itemize}

Thanks to these rules, it follows, depending on the value of parameter $K$, that a direct advantage of this cyclic job assignment is its ability to considerably reduces the number of task migrations.

\subparagraph{Example.} In order to illustrate our claim, consider a platform comprised of $3$ identical $\CPU$s $\pi_1, \pi_2, \pi_3$. We assume that $K = 11$. For a migrating task $\tau_i$, let us consider that $A[i, 1] = 4$, $A[i, 2] = 2$ and $A[i, 3] = 5$, meaning that 4 jobs of $\tau_i$ are assigned to $\CPU$ $\pi_1$, 2 jobs to $\CPU$ $\pi_2$ and 5 jobs to $\CPU$ $\pi_3$. The computations for defining the job assignment sub-sequences using Equation~\ref{eq:sigma} are summarized in Table~\ref{tab:regular}. 

\begin{table}[h]
\footnotesize
\begin{center}
\begin{tabular}{| r | c | c | c | c | c | c | c | c | c | c | c | c |}
\hline
$\ell$ & $0$ & $1$ & $2$ & $3$ & $4$ & $5$ & $6$ & $7$ & $8$ & $9$ & $10$  & total \\ \hline
$\pi_1$ & $1$ & $0$ & $1$ & $0$ & $0$ & $1$ & $0$ & $0$ & $1$ & $0$ & $0$ & $4$ \\ \hline
$\pi_2$ & $1$ & $0$ & $0$ & $0$ & $0$ & $1$ & $0$ & $0$ & $0$ & $0$ & $0$ & $2$ \\ \hline
$\pi_3$ & $1$ & $0$ & $1$ & $0$ & $1$ & $0$ & $1$ & $0$ & $1$ & $0$ & $0$ & $5$ \\ \hline
\end{tabular}
\caption{\label{tab:regular} {\small \bf Most regular-job assignment sequence}}
\end{center}
\end{table}

At step $\ell = 5$ for instance, one job of $\tau_i$ will be assigned to $\pi_1$, then one job to $\pi_2$, and no job will be assigned to $\pi_3$. Whenever $\sigma_k^j = 0$ ($\forall j$), as this is case at steps $1$, $7$ and $10$ in Table~\ref{tab:regular}, then none of the jobs is assigned to none of the $\CPU$s. Such a sub-sequence can be ignored while defining the cyclic assignment controller. In this example, the complete assignment sequence, that will be used to define the cyclic assignment controller for jobs of task $\tau_i$ is derived as follows.   
\begin{eqnarray*}
\sigma &=& (\sigma_0, \sigma_1, \ldots, \sigma_{K - 1})\\
             &=& (\pi_1, \pi_2, \pi_3, \pi_1, \pi_3, \pi_3, \pi_1, \pi_2, \pi_3, \pi_1, \pi_3)
\end{eqnarray*}

While using this strategy, even though this algorithm considerably reduces the number of task migrations, we have a recursion problem since to determine the job-pattern we need actually to know the job-pattern. Indeed, in order to know that $A[i, j]$ jobs of task $\tau_i$ are assigned to $\CPU$ $\pi_j$, we need to perform a schedulability test. However, such a test needs the pattern of the multiframe tasks assigned to each $\CPU$ to be known beforehand. A solution to this recursion problem is to use a schedulability test which requires \emph{only} the number of jobs. That is, the one based on a worst-case scenario, which necessarily introduces a lot of pessimism.

The next job-pattern determination technique will address this drawback and consequently schedules a larger number of task systems.

\subsection*{Alternative job-pattern algorithm}
This cyclic job assignment algorithm has been designed in order to overcome the drawback highlighted in the previous one (see algorithms~\ref{alg:algo2}--\ref{alg:algo1}). It considers each $\CPU$ individually and compute the matrix $A$ incrementally. Initially, $A[i, 1]$ is initialized at $K$. If the obtained multiframe task is not schedulable upon $\pi_1$, then $A[i, 1]$ is decremented to $K-1$. The decrement is repeated until we get the largest number jobs of $\tau_i$ which are schedulable upon $\pi_1$. In the previous example, integers between $11$ and $5$ have been tested without success, but $A[i, 1] = 4$ succeeded, then the multiframe task which will be considered upon this $\CPU$ is $\tau_i^1 = ((C_i, 0, C_i, 0, 0, C_i, 0, 0, C_i, 0, 0), D_i, T_i)$ using \emph{Step~2} of the previous method (Equation~\ref{eq:sigma} in particular). In addition, instead of directly considering a multiframe task with $K$ frames each time, this alternative algorithm allows us to temporarily consider multiframe tasks with a number of frames equal to the number of remaining jobs.

In order to determine the multiframe task $\tau_i^2$ that will be executed upon $\CPU$ $\pi_2$, we temporarily consider a multiframe task $\tau_i^{'2}$ with ($J = K - A[i, 1]$) frames (corresponding to the number of remaining jobs for $\tau_i^2$). This temporal multiframe task $\tau_i^{'2}$ is built by using the \emph{Step~2} of the previous method (Equation~\ref{eq:sigma} in particular), and then we determine its pattern by considering $\tau_i^{'2} \equals ((\sigma_1^2 C_i, \ldots, \sigma_J^2 C_i), D_i, T_i)$ with $\sigma_{\ell}^2 \in \{0,1\}, 1 \leq \ell \leq J$. After this operation has been performed, the multiframe task $\tau_i^2$ is provided with $K$ frames by using $\tau_i^1$ and $\tau_i^{'2}$, based on the following two rules: ($i$) If the $j^{th}$ frame of $\tau_i^1$ is nonzero, then the $j^{th}$ frame of $\tau_i^2$ equals zero. ($ii$) If the $j^{th}$ frame of $\tau_i^1$ is zero and corresponds to the $q^{th}$ zero-frame of $\tau_{i}^1$, then the $j^{th}$ frame of $\tau_i^2$ is determined by using the value of the $q^{th}$ frame of $\tau_i^{'2}$.

These rules can be generalized very easily. Suppose our aim is determining the multiframe task $\tau_i^k$ corresponding to $\tau_i$ upon $\CPU$ $\pi_k$. The steps to perform are as follows. We first determine the amount of remaining jobs $J \equals K - \sum_{q=1}^{k-1} A[i, q]$. Then, we compute $\tau_{i}^{'k}$ thanks to Equation~\ref{eq:sigma}, by using $J$ rather than $K$. Finally, we determine $\tau_i^k$ by using the pseudo-code of Algorithm~\ref{alg:algo2} where $\tau_i^\ell(j)$ denotes the $j^{th}$ frame of $\tau_i^\ell$. That is, if $\tau_i^\ell = ((C_i, 0, C_i), D_i, T_i)$, then $\tau_i^\ell(1) = C_i$, $\tau_i^\ell(2) = 0$ and $\tau_i^\ell(3) = C_i$.

\begin{algorithm}
\caption{\label{alg:algo2} Computation of $\tau_i^k$}

\footnotesize
\KwIn{$K$, $k$, $\tau =(\tau_i^1, \dots, \tau_i^{(k-1)})$, $\tau_i^{'k}$}
\KwOut{Multiframe task $\tau_i^k$}
\SetKwFunction{ComputeT}{Compute}
{\bf Function} \ComputeT($\operatorname{in }\tau_i^{'k}, \operatorname{in }K, \operatorname{in }\tau$) \\
\Begin{
	$q \gets 1$ \;
	\For{(j = 1 \ldots K)}{
    		$\Free \gets \True$\;
    		$\ell \gets 1$\;
    		\While{($\ell < k $ {\em {\bf and}} $\Free$)}{
        			\lIf{ $\tau_i^\ell(j) = 0$ }{
           			$\ell \gets \ell + 1$\;
        			}\lElse{
           			$\Free \gets \False$\;
        			}
    		}
    		\If{ $\Free$ }{
       			$\tau_i^k(j) \gets \tau_i^{'k}(q)$\;
       			$q \gets q+1$\;
    		}\lElse{
       			$\tau_i^k(j) = 0$\;
    		}    
  	}	
\Return{$\tau_i^k$} \;
}
\end{algorithm}

\begin{algorithm}
   \footnotesize
   \caption{ \label{alg:algo1} Jobs assignment for migrating task $\tau_i$}
   \KwIn{$K$, task $\tau_i$}
   \KwOut{$\True$ and $A$ if $\tau_i$ is schedulable, $\False$ otherwise}   
   \SetKwFunction{AlgoTwo}{Algo2}
   {\bf Function} \AlgoTwo($\operatorname{in }\tau_i$, $\operatorname{in }K, \operatorname{out }A$) \\
   \Begin{
      $\RemainingJobs \gets K$ ; {\em /* Remaining jobs*/} \\
      \For{($k = 1 \cdots m$)}
      {
         \For{($j = \RemainingJobs \cdots 1$)}
         {
            Compute $\tau_i^{'k}$ by using Equation~\ref{eq:sigma} \;
            $\tau_i^k = \ComputeT(\tau_i^{'k}, K, (\tau_i^1, \dots, \tau_i^{(k-1)}))$\;
            \If{($\tau_i^k \operatorname{\:\: is \:\: schedulable \:\: on \:\: \pi_k}$)}
            {
               $A[i,k] \gets j$\;
               $\RemainingJobs \gets \RemainingJobs - j$ \;            
               Exit\;
            }
         }
         \lIf {$\RemainingJobs = 0$} \Return $\True$ \;
      }
      \Return $\False$
   }
  \end{algorithm}

\subparagraph{Example.} Consider the same task $\tau_i$ as in the previous example. Assuming only $4$ jobs can be assigned to $\pi_1$, then the multiframe task $\tau_i^1$ upon $\pi_1$ is $\tau_i^1 = ((C_i, 0, C_i, 0, 0, C_i, 0, 0, C_i, 0, 0), D_i, T_i)$ by using Algorithm~\ref{alg:algo2}. If $\tau_i^1$ is schedulable upon $\pi_1$, then we consider $\pi_2$ and we repeat the same process. If not, we try to assign jobs to $\pi_1$, now assuming only $3$ jobs can be assigned this time.

In this example, we assume that $4$ jobs have successfully been assigned to $\pi_1$. We now assume that neither $4$, nor  $3$ jobs of $\tau_i$ cannot be assigned to $\pi_2$, but $2$ jobs can. Computing the intermediate task $\tau_i^{'2}$ using $J = 11 - 4 = 7$ leads us to $\tau_i^{'2} = ((C_i, 0, 0, C_i, 0, 0, 0), D_i, T_i)$. By applying Algorithm~\ref{alg:algo2} again, we obtain $\tau_i^2 = ((0, C_i, 0, 0, 0, 0, C_i, 0, 0, 0, 0), D_i, T_i)$. 

We repeat the same process for $\CPU$ $\pi_3$. When trying to assign the $5$ remaining jobs, computing $\tau_i^{'3}$ with $J = 5$ gives $\tau_i^{'3} = ((C_i, C_i, C_i, C_i, C_i), D_i, T_i)$. Again, applying Algorithm~\ref{alg:algo2}, we obtain $\tau_i^3 = ((0, 0, 0, C_i, C_i, 0, 0, C_i, 0, C_i, C_i), D_i, T_i)$.

\subsection{Semi-partitioning scheduling algorithm}\label{subsec:scheduling algorithm} 

This section describes our semi-partitioning algorithm based on the concept of semi-partitioned scheduling. Like traditional partitioned scheduling algorithms, the schedulers have the same scheduling policy upon each $\CPU$, that is $\EDF$ in this case, but each of them is not completely independent, because several $\CPU$s may share a migrating task. The principle of the algorithm is as follow.
\begin{itemize}
\item Sort the tasks in decreasing order of their utilization factor.

\item For each task considered individually:

	\begin{itemize}
	\item Select one $\CPU$ and set it as the one with the smallest index.

	\item If the current task is schedulable upon the selected $\CPU$ while using the $\FFD$ algorithm (i.e. all jobs meet their 	deadlines), then assign the task to this $\CPU$.

	\item If the current task is neither schedulable upon the selected $\CPU$ nor upon the others by using the $\FFD$ algorithm, then determine the number of jobs that can be scheduled upon the selected $\CPU$ w.r.t.\@ an $\EDF$ scheduler. Assign those jobs to the selected $\CPU$ by using the multiframe task approach defined earlier. Then, move to the next $\CPU$. Repeat this process until all the jobs of the task are assigned to a $\CPU$.
	\end{itemize}
\end{itemize}

In the end, the intuitive idea behind our algorithm is to consider one task at each step and defines the minimum number of $\CPU$s required to execute all its jobs (this takes at most $m$ iterations, where $m$ is the number of $\CPU$s). At the beginning, the algorithm performs using an $\FFD$ algorithm. Then, whenever a task cannot be assigned to a particular $\CPU$, its jobs are assigned to not fully utilized $\CPU$s (at most $m$). Thus, the algorithm requires at most $\mathcal{O}(nm)$ schedulability tests.

\begin{algorithm}[t]
\footnotesize
\KwIn{$m$, $\tau = \{\tau_1, \ldots, \tau_n\}$, parameter $K$.}
\KwOut{$\True$ and $A$ if $\tau$ is schedulable, $\False$ otherwise.}
\SetKwFunction{AlgoSemiPart}{SemiPart}
{\bf Function} \AlgoSemiPart($\operatorname{in }m, \operatorname{in }K, \operatorname{in }\tau, \operatorname{out }A$) \\
\Begin{
   \For{(i=1 \ldots n)}{  
      $\Sched \gets \False$\;
       {\em /*We try to schedule the task using $\FFD$ */}\\
       \For{(j=1 \ldots m)}{
          \If{$\tau_i\operatorname{\:\: schedulable \:\: on \:\:}\pi_j$}{
             $\tau_i$ is assigned to $\pi_j$\;
             $\Sched \gets  \True$\;
          }
       }
		\If{($\Sched == \False$)}{
         {\em /*Task $\tau_i$ is a migrating task. We perform \\ the job assignment among the $\CPU$s by \\ calling Agorithm~\ref{alg:algo1} */}\\
         \If{$\AlgoTwo(\tau_i, K, A) == \False$}
         {
            {\em /*$\tau_i$ is not schedulable */} \\
            \Return $\False$\;
		  	}
      }
   }
   \Return $\True$\;
}
\caption{Semi-Partitioning Algorithm}
\label{alg:split}
\end{algorithm}

\section{Schedulability analysis}\label{sec:schedulability analysis} 

This section derives the schedulability conditions for our algorithm. Because each migrating task is modeled as a multiframe task on each $\CPU$ upon which it has a job assigned to, the schedulability analysis is performed on each $\CPU$ individually, using results from the uniprocessor scheduling theory. As tasks are scheduled upon each $\CPU$ according to an $\EDF$ scheduler, a specific analysis is necessary only for $\CPU$s executing at least one multiframe task. For such a $\CPU$, classical schedulability analysis approaches such as ``{\em Processor Demand Analysis}'' cannot applied, unfortunately. This is due to migrating tasks. We consider two scenarios to circumvent this issue: the {\em packed scenario} and the {\em scenario with pattern} and we develop a specific analysis based on an extension of the {\em Demand Bound Function} ($\DBF$) \cite{Baker07, Baruah1999} to multiframe tasks. We recall that a natural idea to cope with each migrating task is to duplicate it as many time as the number of $\CPU$s it has a job assigned to and consider a multiframe task upon each $\CPU$.

\subsection{Packed scenario}
\label{sec:packed_scenario}

As each migrating task $\tau_i$ is modeled as a collection of at most $m$ multiframe tasks $\tau_i^1$, $\tau_i^2$, \ldots, $\tau_i^m$, where $\tau_i^j$ is to be executed upon $\CPU$ $\pi_j$ ($m$ is the number of $\CPU$s), this scenario considers the {\em worst-case}. This occurs when, for each multiframe task associated to $\tau_i$, all nonzero execution requirements are at the beginning of the frames and the remaining execution requirements are set to zero. That is $\tau_i^k$ is modeled as $\tau_i^k \equals ((C_i, C_i, \ldots, C_i, 0, \ldots, 0), D_i, T_i)$, where $C_i$ is the execution requirement of task $\tau_i$. (see Appendix for the proof). 

In order to define $\widehat{\DBF}$ at time $t$ (that is, our adapted Demand Bound Function for systems such as those considered in this paper), we need to define the contribution of each task in the time interval $[0,t)$. As $K$ denote the number of execution requirements in a multiframe task $\tau_i^k$ associated to $\tau_i$, let $\ell_i^k$ denote the number of nonzero execution requirements in $\tau_i^k$. When using the packed scenario, $\tau_i^k$ can be modeled as $\tau_i^k = ((\underbrace{C_{i}, \ldots, C_{i}}_{\ell_i^k\text{ elements}}, 0, \ldots, 0), D_i, T_i)$. The challenge is to take into account that only $\ell_i^k$ jobs of task $\tau_i^k$ will contribute to the $\DBF$ in the time interval $[0, K \cdot T_i)$. 

\begin{figure*}
\begin{center}
	\includegraphics[width=\linewidth]{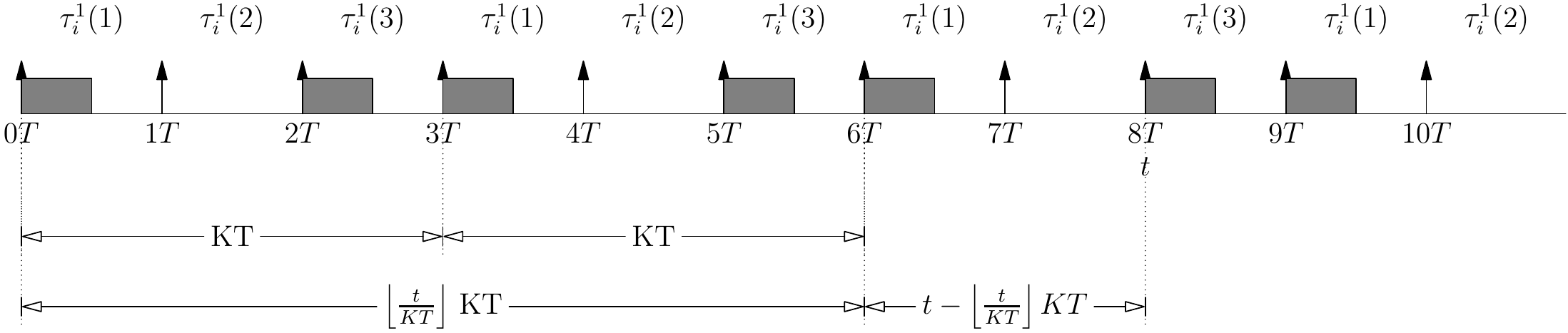}
	\caption{Illustration of the multiframe task $\tau_i^1 = ((C_i, 0, C_i), T_i, T_i)$.}
	\label{fig:dbf}
\end{center}
\end{figure*}

Based on Figure~\ref{fig:dbf}, the contribution to $\widehat{\DBF}$ at time $t$ for the multiframe task $\tau_i^k$ is determine as follows.

\begin{itemize}
\item First, consider the number of intervals of length $K \cdot T_i$ by time $t$. Letting $s$ denote that number, we have: 
\begin{equation}\label{eq:s}
s \equals \left\lfloor \frac{t}{K \cdot T_i} \right\rfloor
\end{equation}
As such, the contribution of $\tau_i^k$ to the $\DBF$ is at least $s \cdot \ell_i \cdot C_i$.

\item Second, consider the time interval $[s \cdot K \cdot T_i, t)$, where several jobs may have their deadlines before time $t$. Assuming that all the jobs are assigned to the considered $\CPU$, the number of jobs (``$a$'') is given by:
\begin{equation}\label{eq:a}
a \equals \max \left(0,\left\lfloor \frac{(t \bmod{K \cdot T_i}) - D_i}{T_i}\right\rfloor+1\right)
\end{equation}
Since at most $\ell_i^k$ jobs over $K$ will be executed upon that $\CPU$, then the exact contribution of $\tau_i$ in interval $[s \cdot K \cdot T_i, t)$ is given by: $\min(\ell_i^k, a) \cdot C_i$.
\end{itemize}
Putting all of these expressions together, the $\DBF$ of $\tau_i^k$ at time $t$ is defined as follows.
\begin{equation}\label{eq:dbf_packed}
\widehat{\DBF}(\tau_i^k,t) \equals s \cdot \ell_i^k \cdot C_i  + 
                          \min\left(\ell_i^k, a\right) \cdot C_i
\end{equation}

\paragraph{Schedulability Test $1$.} Let $\tau = \{\tau_1, \tau_2, \ldots, \tau_n\}$ be a set of $n$ constrained-deadline sporadic tasks to be scheduled upon an identical multiprocessor platform. For each selected $\CPU$ $\pi$, if $\mathcal{N} \subseteq \tau$ denotes the subset of {\em non-migrating tasks} upon $\pi$, then a {\em sufficient condition} for the system to be schedulable upon $\pi$ is given by 
\begin{equation}\label{condi}
\sum_{\tau_j \in \mathcal{N}} \DBF(\tau_j, t)  + \sum_{\tau_i^k \in \pi \backslash \mathcal{N}} \widehat{\DBF}(\tau_i^k, t) \leq t \quad \forall t
\end{equation}
In Equation~\ref{condi}, $\tau_j$ is a {\em non-migrating task}, $\DBF$ is the classical Demand Bound Function and $\tau_i^k$ is a multiframe task assigned to $\pi$.

\subsection{Scenario with pattern}\label{subsec:scenario with pattern} 

This scenario, in contrast to the packed one which considers the worst-case, takes the pattern of job assignment provided by our algorithm into account. Indeed, implementations showed that the schedulability analysis based on the packed scenario may be too pessimistic, unfortunately. 

The schedulability test developed in this section is similar to the one developed previously for packed scenarios in that it is also based on the $\DBF$. However the pattern of job assignment to $\CPU$s is taken into account. Therefore, the only noticeable difference comes from the second term of Equation~\ref{eq:dbf_packed} which is replaced by:
\begin{equation}\label{caca}
   \max_{c=0}^{K-1} \left(\sum_{j=c}^{c + nb_i(t)-1} C_{i, j \bmod K}\right)
\end{equation}
In Expression~\ref{caca}, $nb_i(t) \equals \left\lfloor \frac{(t \bmod K \cdot T_i)- D_i}{T_i} \right\rfloor + 1$ and denotes the number of jobs of $\tau_i^k$ in the time interval $[s \cdot K \cdot T_i, t)$. This is, we compute the processor demand by considering that the ``critical instant'' coincide with the first job of the pattern, then the second and so on, until the $K^{th}$ job. After that, we consider the maximum processor demand in order to capture the worst-case. Hence we obtain: 

\begin{equation}\label{eq:dbf_pattern}
\widehat{\widehat{\DBF}}(\tau_i,t) \equals s \cdot \ell_i^k \cdot C_i  +  \max_{c=0}^{K-1} \left(\sum_{j=c}^{c + nb_i(t)-1} C_{i, j \bmod K}\right)
\end{equation}

\paragraph{Schedulability Test $2$.} Let $\tau = \{\tau_1, \tau_2, \ldots, \tau_n\}$ be a set of $n$ constrained-deadline sporadic tasks to be scheduled upon an identical multiprocessor platform. For each selected $\CPU$ $\pi$, the {\em sufficient condition} for the system to be schedulable upon $\pi$ is as the previous one except that in Expression~\ref{condi} the value of $\widehat{\DBF}$ is now replaced by $\widehat{\widehat{\DBF}}$.

Note that Lemma~\ref{lemma} in the Appendix provides a proof of the ``{\em dominance}'' of Schedulability Test~$2$ over Schedulability Test~$1$. That is, all systems that are schedulable using the Test~$1$ are also schedulable using Test~$2$. 
\section{Experimental results}\label{Experimental results}

In this section, we report on the results of experiments conducted using the theoretical results presented in Sections~\ref{sec:algo_principle}~and~\ref{sec:schedulability analysis}. These experiments help us to evaluate the performances our algorithms relative to both the $\FFD$ and the S. Kato algorithms. Moreover, they help us to point out the influence of some parameters such as the value of parameter $K$ and the number of $\CPU$s in the platform. We performed a statistical analysis based on the following characteristics: ($i$) the number of $\CPU$s is chosen in the set $M \equals \{2, 4, 8, 16, 32, 64\}$ from practical purpose. Indeed multicore platforms often have a number of cores which is a power of $2$, ($ii$) the system utilization factor for each $\CPU$ varies between $0.50$ and $0.95$, using a step of $0.05$.

During the simulations, $10.000$ runs have been performed for each configuration of the pair (number of $\CPU$s in the platform, system utilization factor). In the figures displayed below, a ``{\em success ratio}'' of $2\%$ of an algorithm $\mathcal{A}$ over an algorithm $\mathcal{B}$ means that $\mathcal{A}$ leads to a schedulability ratio of $y\%$, where $\mathcal{B}$ leads to a schedulability ratio $(y-2)\%$. 

For comparison reasons, the same protocol as the one described in~\cite{Kato2009} by S. Kato for generating task systems has been considered. That is:
\begin{enumerate}
  \item the utilization factor of each task is randomly generated in the range $[0, 1]$ with the constraint that the sum of utilization factors of all tasks must be equal to the utilization factor of the whole system,
  
  \item the period $T_i$ of task $\tau_i$ is randomly chosen in the range $[100, 3000]$,
  
  \item the relative deadline $D_i$ of task $\tau_i$ is set equal to its period,
  
  \item the $\WCET$ $C_i$ of task $\tau_i$ is computed from its period $T_i$ and its utilization factor $u_i$ as $C_i \equals u_i \cdot T_i$.
\end{enumerate}

It is worth noticing that we have no control of the number of tasks composing the system. This number depends on the utilization factors of the tasks. 

Figure~\ref{fig:packed_K} depicts the curve of the success ratio relative to parameter $K$ for ``{\em packed scenarios}''. As we can see, the lower the value of $K$, the higher the success ratio. This result may be surprising and counter-intuitive at first glance as we would expect the contrary, that is, increasing $K$ would improve the success ratio. A reason to this is provided by the pessimism introduced when packing all the nonzero frames at the beginning of each multiframe task. Indeed in this case, the processor demand is over-estimated at time $t = 0$ for each $\CPU$ upon which there is a multiframe task. To illustrate our claim, when $K=2$, each migrating task is assigned to at most $2$ $\CPU$s, thus increasing the pessimism for these $\CPU$s. Consequently, the larger the value of $K$, the larger the number of $\CPU$s for which the processor demand will be over-estimated. 

\begin{figure*}
\centering
\begin{tabular}{ccc}
\epsfig{file=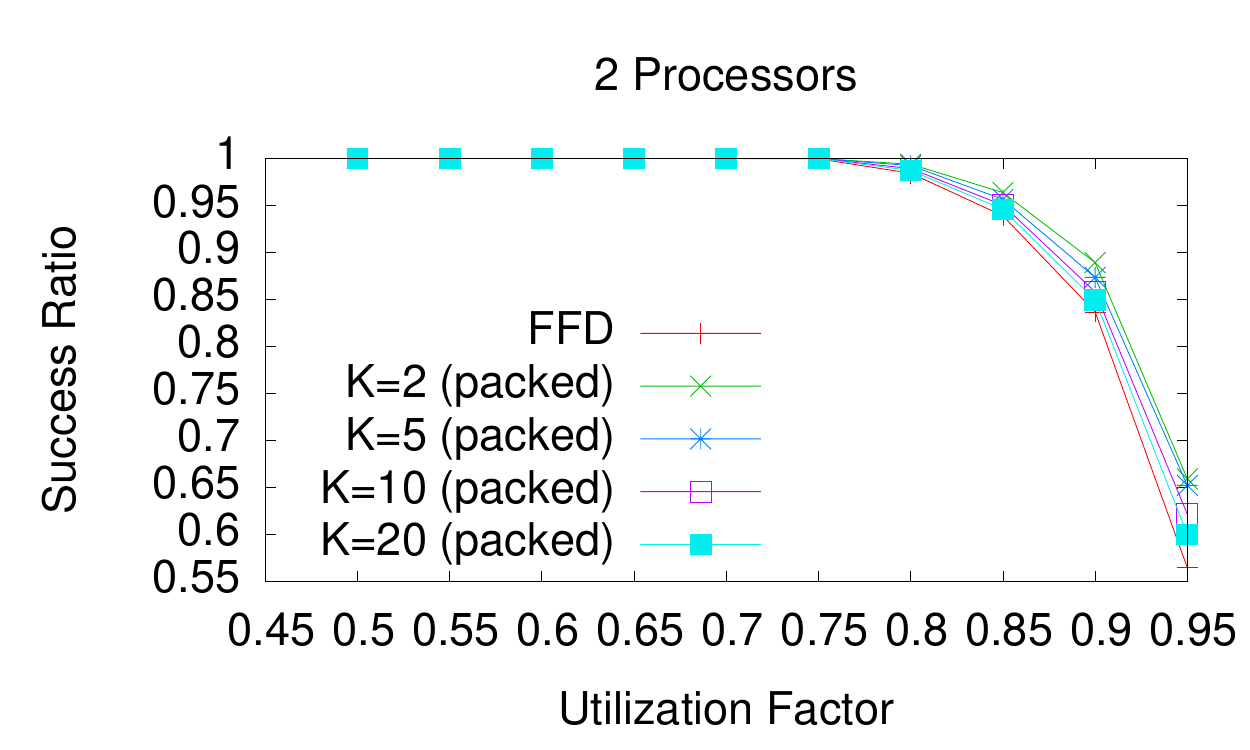,width=0.31\textwidth, height=0.22\textheight} & 
\epsfig{file=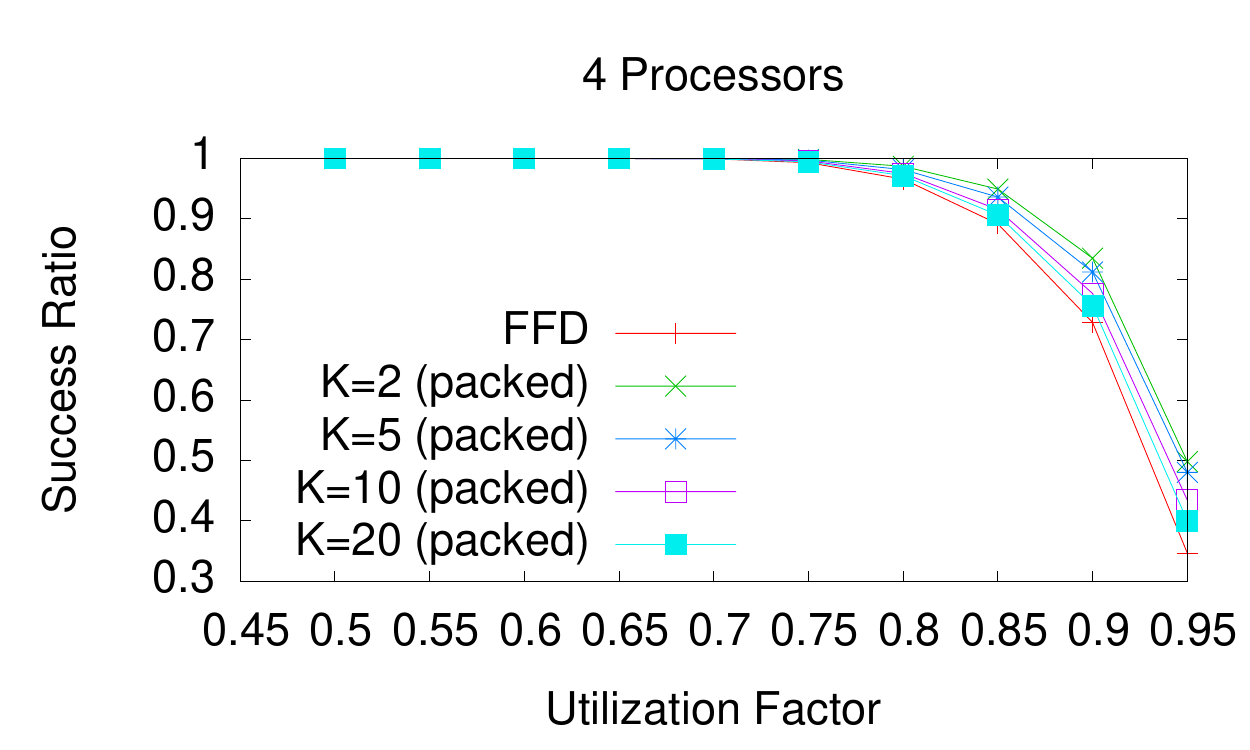,width=0.31\textwidth, height=0.22\textheight} &
\epsfig{file=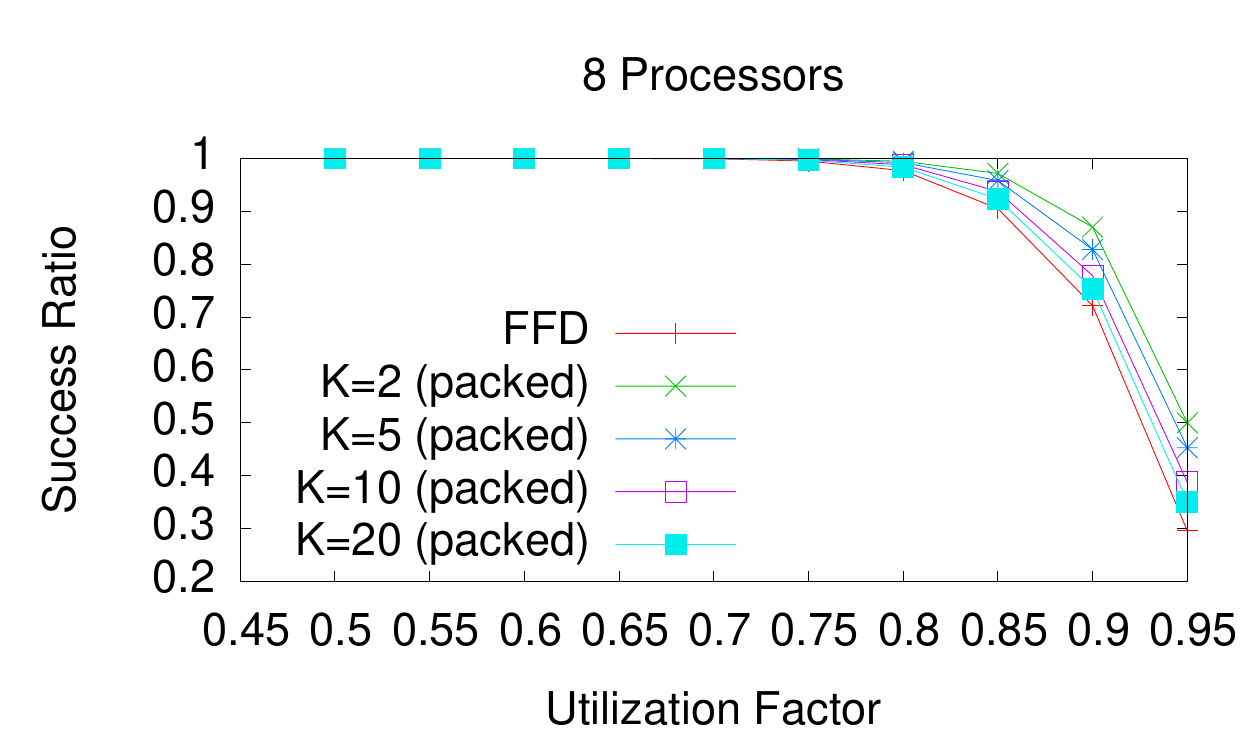,width=0.31\textwidth, height=0.22\textheight}
\end{tabular}
\caption{Success ratio relative to parameter K for packed scenarios}
\label{fig:packed_K}
\end{figure*}

Figure~\ref{fig:pattern_K} depicts the curve of the success ratio relative to parameter $K$ for the ``{\em scenario with pattern}''. Here, in contrast to the results obtained for ``{\em packed scenarios}'', the curves behave as we would intuitively expect them. That is, increasing the value of $K$ leads to a higher success ratio. This is due to the fact that a higher value of $K$ allows the jobs of a migrating task to be assigned to a larger set of $\CPU$s, thus reducing the  global computation requirement upon each $\CPU$.

\begin{figure*}
\centering
\begin{tabular}{ccc}
\epsfig{file=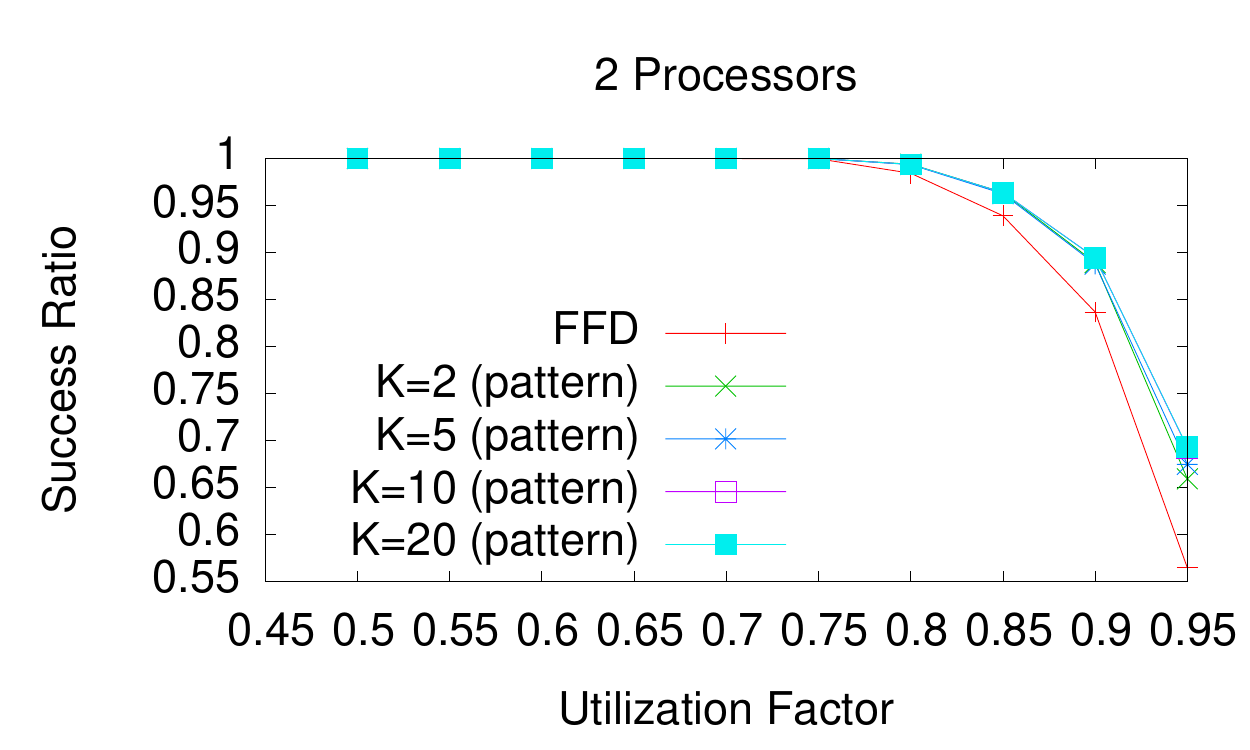,width=0.31\textwidth, height=0.22\textheight} & 
\epsfig{file=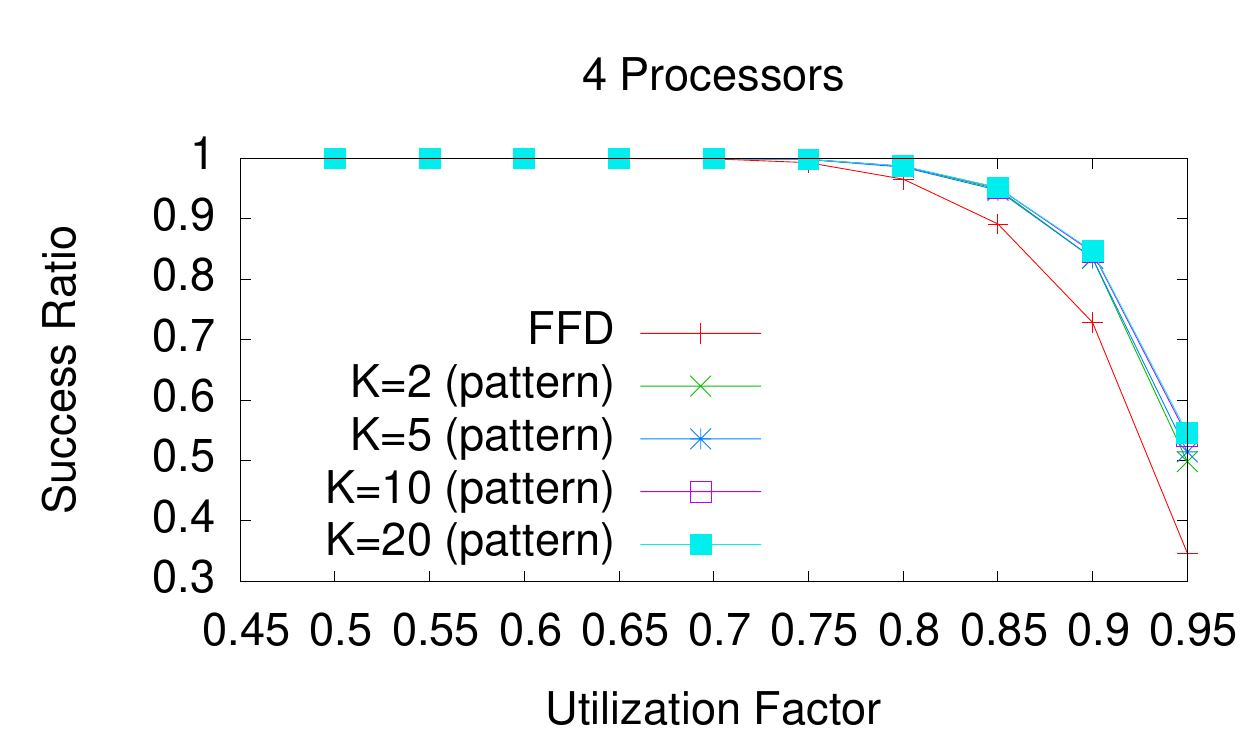,width=0.31\textwidth, height=0.22\textheight} &
\epsfig{file=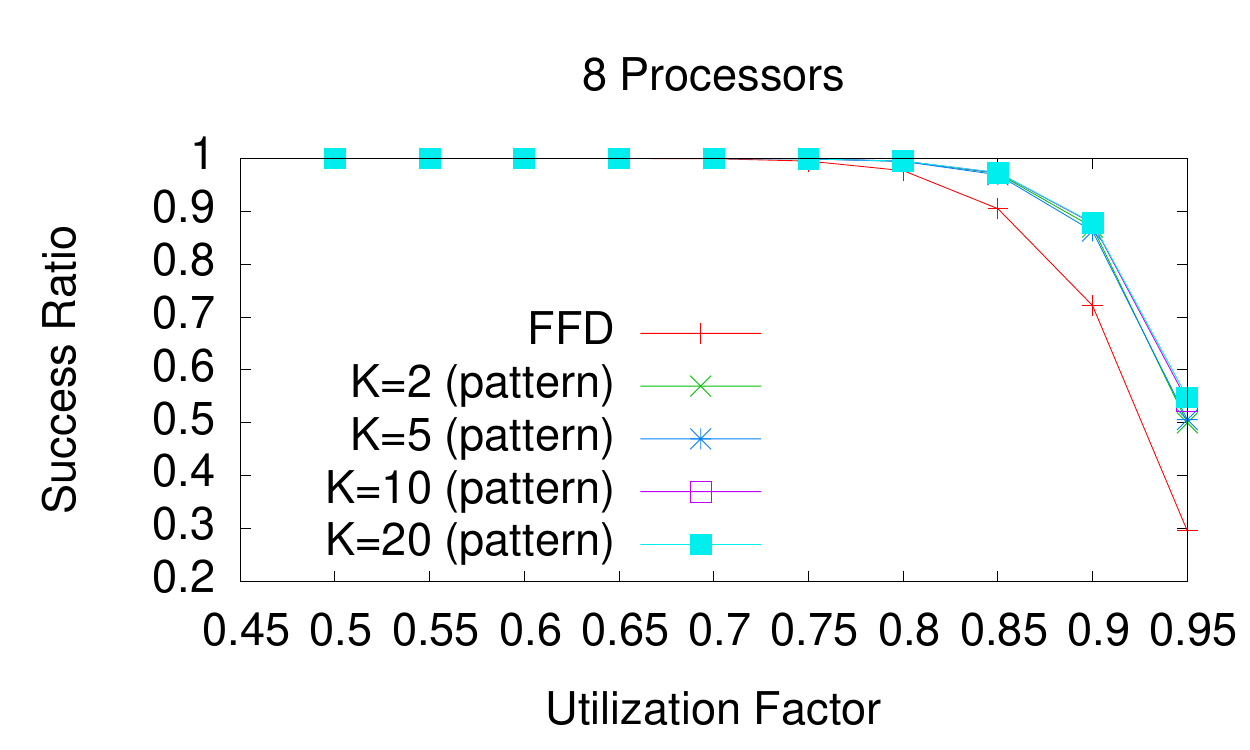,width=0.31\textwidth, height=0.22\textheight}
\end{tabular}
\caption{Success ratio relative to parameter K for scenarios with pattern}
\label{fig:pattern_K}
\end{figure*}

Now, it is worth noticing that both the packed scenario and the scenario with pattern lead to the same algorithm when $K=2$. In fact, while considering a scenario with pattern, the multiframe execution requirements which are allowed for a migrating task $\tau_i$ are $(C_i, 0)$ and $(0, C_i)$. Hence, from the schedulability analysis view point, the one that will be considered is $(C_i, 0)$, which corresponds to the packed scenario.
 
Figure~\ref{fig:comparaison} compares the performances of the algorithms which provided the best results using our approach during the simulations, that is, the one using packed scenarios when $K=2$ and the one using scenarios with pattern when $K=20$, and those of both the $\FFD$ and the S. Kato algorithms. 

\begin{figure*}
\centering
\begin{tabular}{ccc}
\epsfig{file=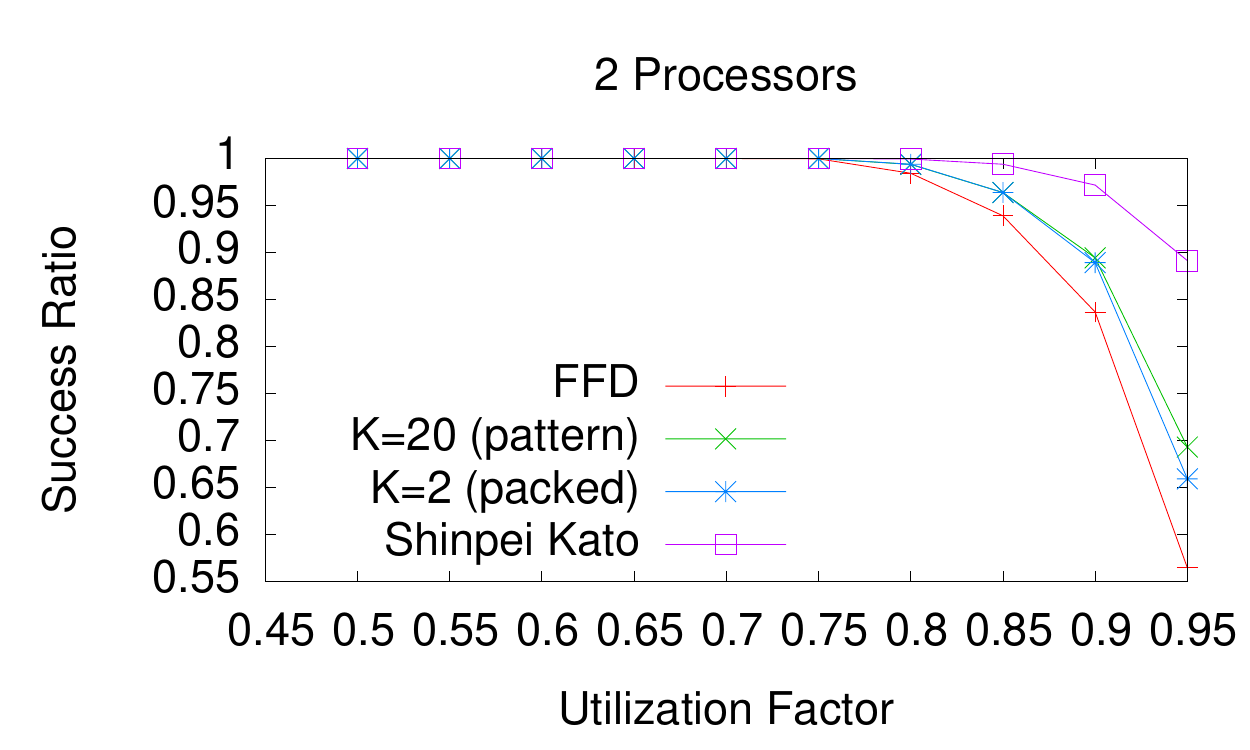,width=0.31\textwidth, height=0.22\textheight} & 
\epsfig{file=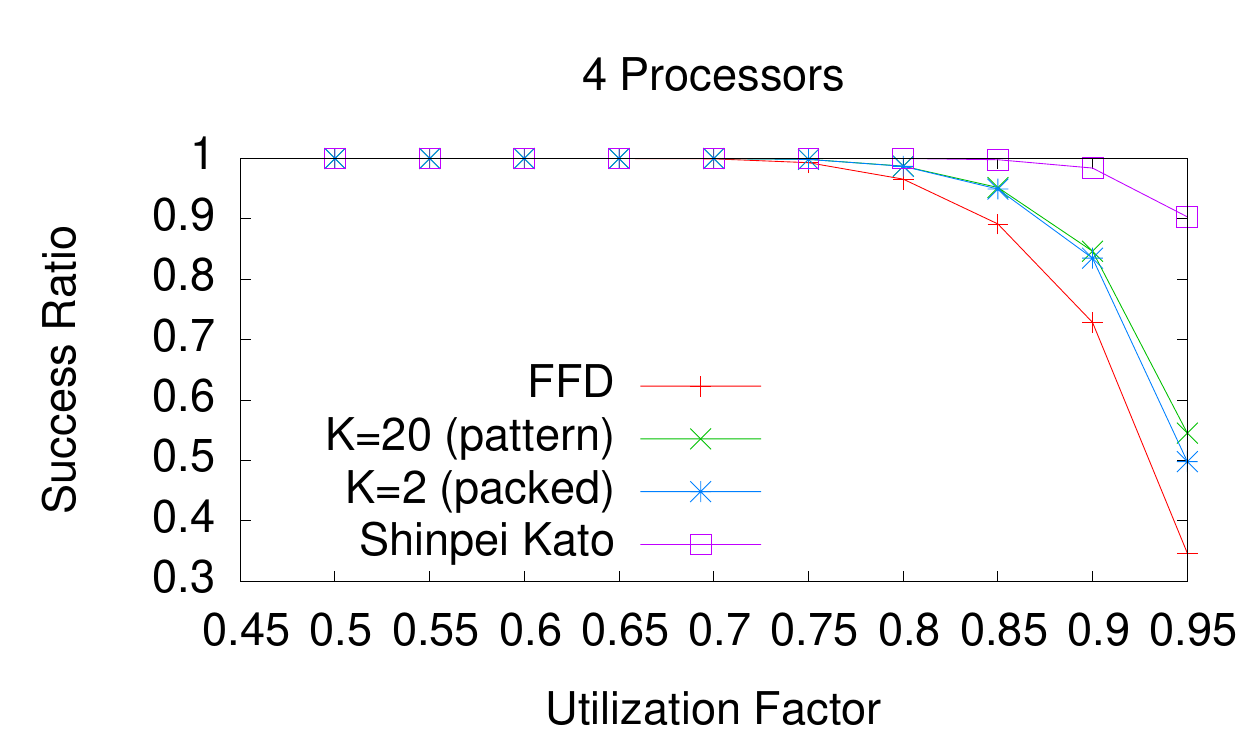,width=0.31\textwidth, height=0.22\textheight} &
\epsfig{file=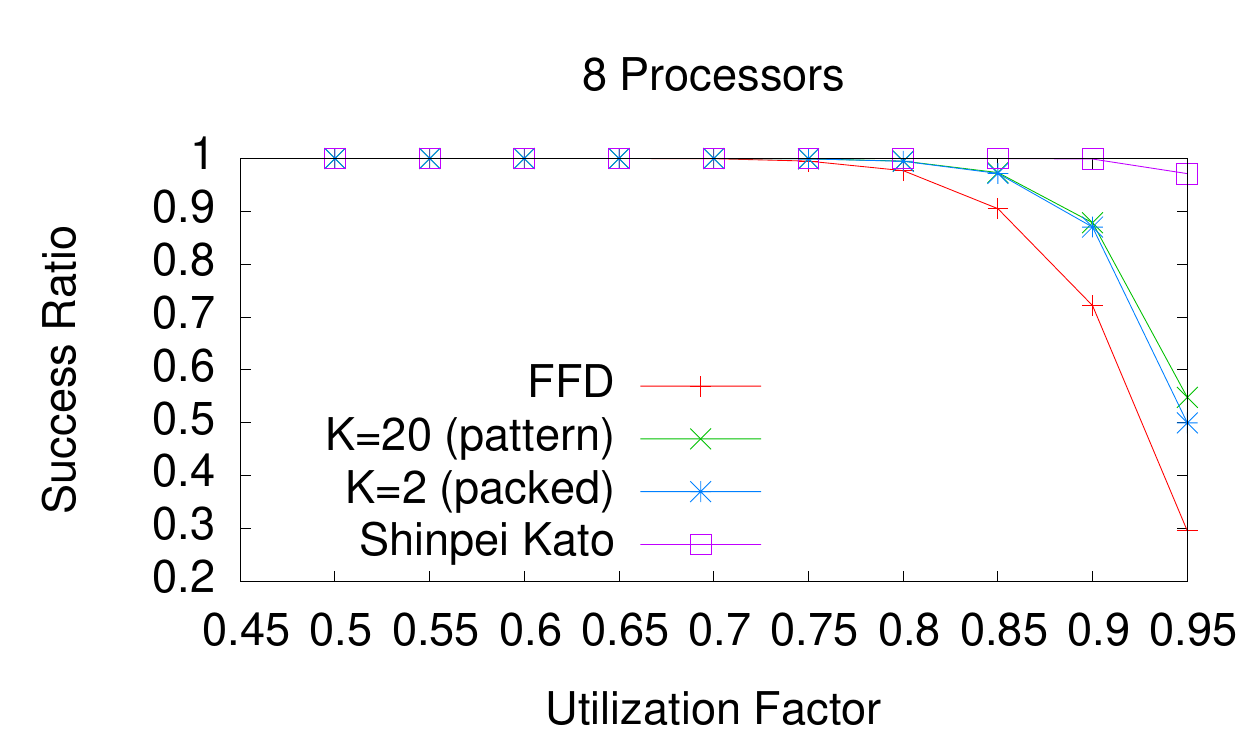,width=0.31\textwidth, height=0.22\textheight} \\
\epsfig{file=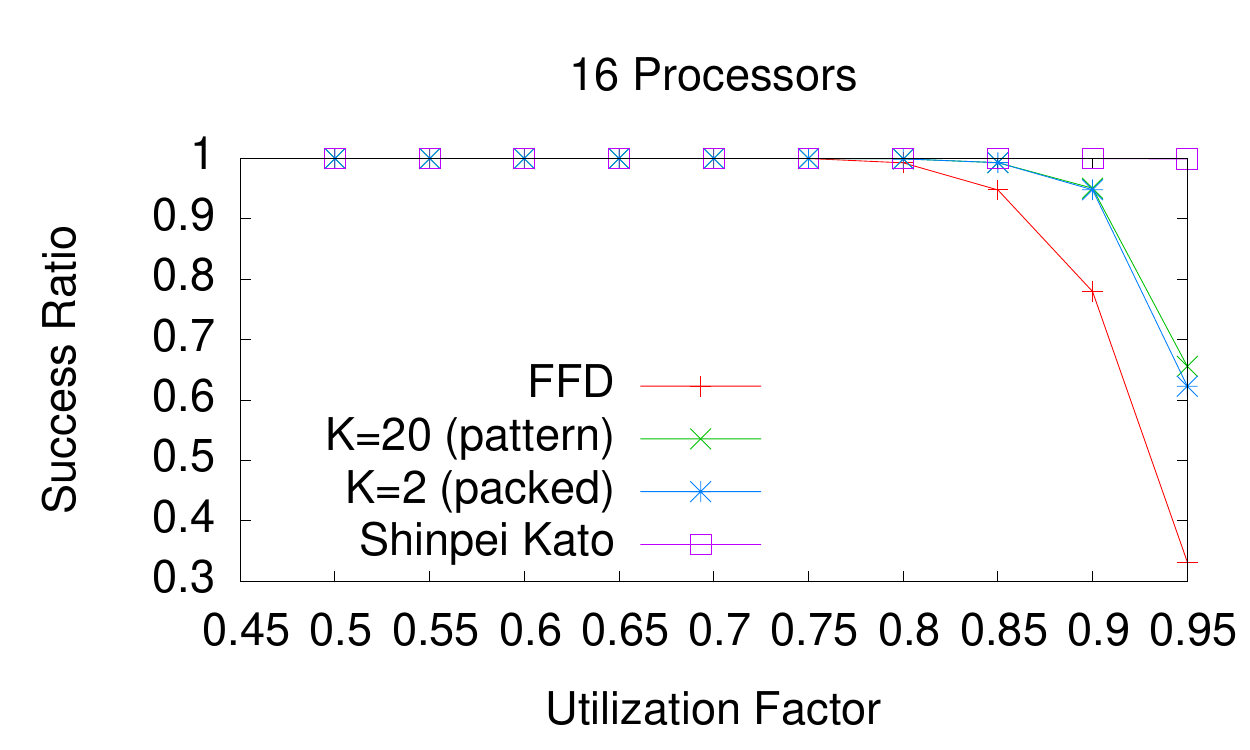,width=0.31\textwidth, height=0.22\textheight} & 
\epsfig{file=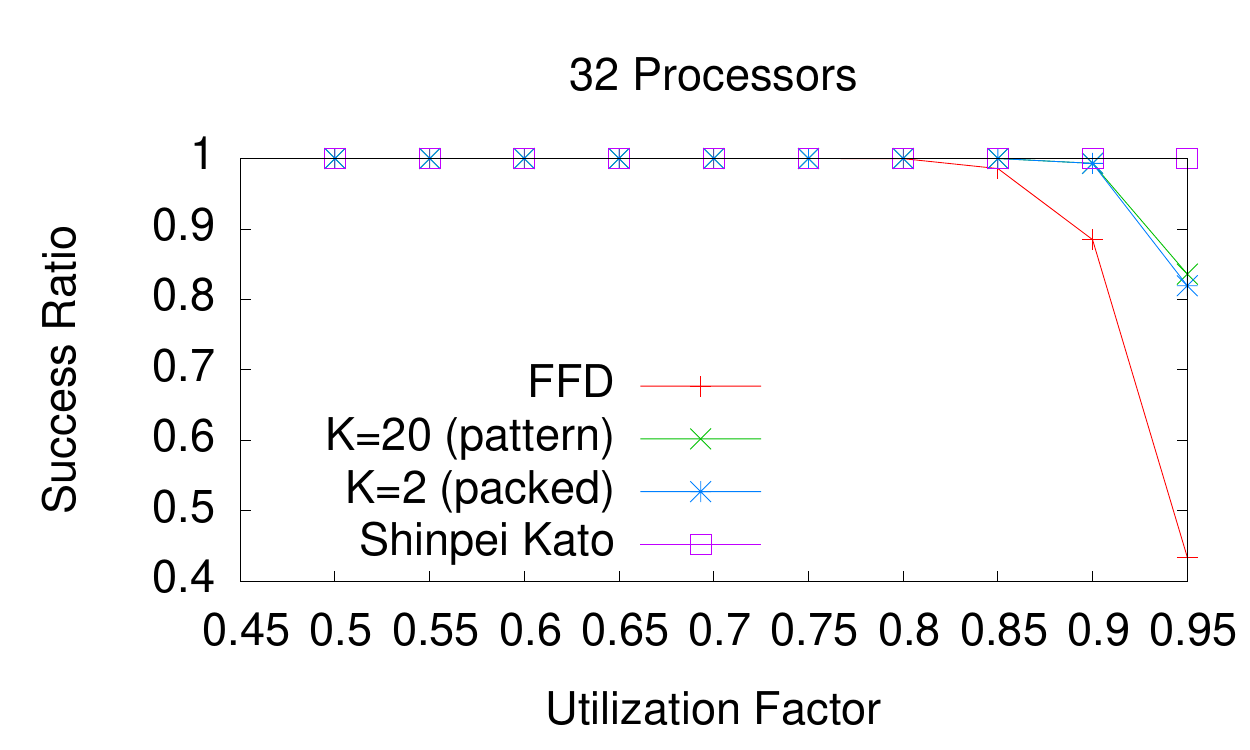,width=0.31\textwidth, height=0.22\textheight} &
\epsfig{file=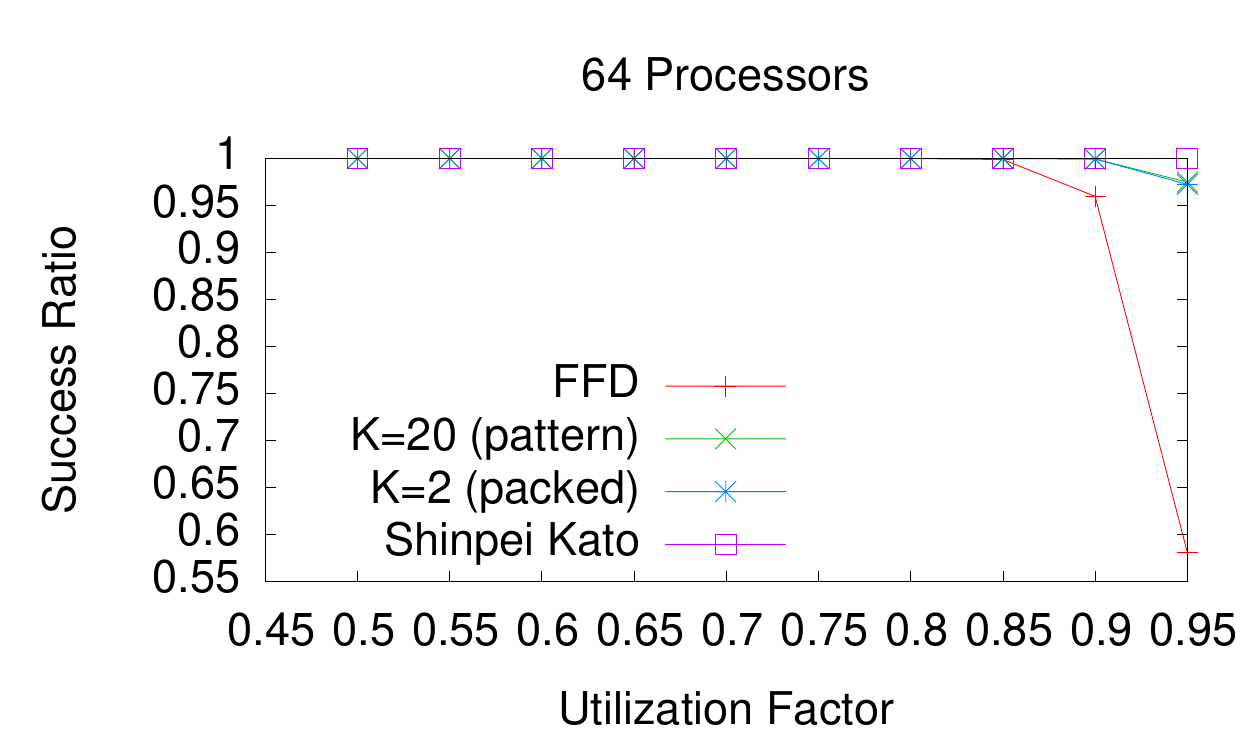,width=0.31\textwidth, height=0.22\textheight}
\end{tabular}
\caption{Comparison of $\FFD$, packed- (K=2), pattern- (K=20) and S. Kato's algorithms}
\label{fig:comparaison}
\end{figure*}

We can see that our algorithms always outperform the $\FFD$ algorithm and are a bit behind the one proposed by S. Kato in terms of success ratio. However, the results provided by our algorithms strongly challenge those produced by the S. Kato one for large number of $\CPU$s. Note that our algorithms in contrast to the S. Kato one do not allow job migrations, thus limiting runtime overheads which may be prohibitive for the system. Hence, we believe that our approach is a promising path to go for more competitive algorithms and for practical use.

\section{Conclusion and Future work}\label{Conclusion and future work}

In this paper, the scheduling problem of hard real-time systems comprised of constrained-deadline sporadic tasks upon identical multiprocessor platforms is studied.  A new algorithm and a scheduling paradigm based on the concept of semi-partitioned scheduling with restricted migrations has been presented together with its schedulability analysis. The effectiveness of our algorithm has been validated by several sets of simulations, showing that it strongly challenges the performances of the one proposed by S.~Kato. Future work will address two issues. The first issue is relaxing the constraint on parameter $K$ as we think that is possible to define a value $K$ for each task such that the number of frames are kept as small as possible. The second issue is solving the optimization problem taking into account the number of migrations. \\ 

\noindent {\bf Acknowledgment.}
The authors would like to thank Ahmed Rahni from LISI/ENSMA for his insightful comments on the schedulability tests of multiframe tasks.


\bibliographystyle{acm}
\bibliography{biblio-semi}

\appendix
\begin{appendix}
\section{Appendix}

Let $\tau_i$ be a multiframe task with a pattern $\sigma \equals (\sigma_1, \ldots, \sigma_K)$, that is, $\tau_i = ((\sigma_1 C_i, \dots, \sigma_K C_i), D_i, T_i)$ with $\sigma_{j} \in \{0,1\}, 1 \leq j \leq K$. Let $\ell_i$ denotes the number of nonzero execution requirements. Let $\tau_i^{(p)}$ be the packed version of $\tau_i$ that is: $\sigma_{j} = 1$ if $j \leq \ell_i$ and $\sigma_{j} = 0$, otherwise.

\begin{lemma}\label{lemma}
   \begin{equation}
      \widehat{\widehat{\DBF}}(\tau_i, t) \leq \widehat{\widehat{\DBF}}(\tau_i^{(p)}, t)
   \end{equation}
\end{lemma}
\begin{proof}   
   By definition:
   \[
   \widehat{\widehat{\DBF}}(\tau_i,t) = s \cdot \ell_i \cdot C_i  + 
                       \max_{c=0}^{K-1} \left(\sum_{j=c}^{c + nb_i(t) - 1} C_{i, j \bmod K}\right)
   \]
In this Equation, the $\max$ term can be rewritten as: 
\begin{equation}
\footnotesize
   \max_{c=0}^{K-1} \left(\sum_{j=c}^{c + nb_i(t) - 1} C_{i, j \bmod K}\right) = \max_{c=0}^{K-1} \left( C_i \sum_{j=c}^{c + nb_i(t) - 1} \sigma_{i, j \bmod K}\right)
\end{equation}
Note that since $\sigma_{j} \in \{0, 1\}$ and we perform the sum of $nb_i(t)$ terms. 
\begin{equation}
   \sum_{j=c}^{c + nb_i(t) - 1} \sigma_{j \bmod K} \leq \min(\ell_i, \max(0, nb_i(t)))
\end{equation}
Because we have at most $\ell_i$ nonzero execution requirements, so, we take the minimum between $\ell_i$ and $nb_i(t)$. As such,
\begin{equation}
\footnotesize
   \max_{c=0}^{K-1} \left(\sum_{j=c}^{c + nb_i(t) - 1} C_{i, j \bmod K}\right) \leq \max_{c=0}^{K-1} \left(C_i \min(\ell_i, \max(0, nb_i(t)))\right)
\end{equation}
leading to
\begin{equation}
\footnotesize
   \max_{c=0}^{K-1} \left(\sum_{j=c}^{c + nb_i(t) - 1} C_{i, j \bmod K}\right) \leq C_i \min(\ell_i, \max(0, nb_i(t)))
\end{equation}
Finally, $\widehat{\widehat{\DBF}}(\tau_i, t) \leq s \cdot \ell_i \cdot C_i + C_i \min(\ell_i, \max(0, nb_i(t)))$.
That is, $\widehat{\widehat{\DBF}}(\tau_i, t) \leq \widehat{\widehat{\DBF}}(\tau_i^{(p)}, t)$. The lemma follows.
\fin
\end{proof}

\begin{theorem}
   Let $\tau_i$ be a multiframe task. The worst-case pattern for $\tau_i$ is the packed pattern.
\end{theorem}

\begin{proof}
A pattern $\mathcal{A}$ is said worse than a pattern $\mathcal{B}$ if and only if $\mathcal{A}$ schedulable implies $\mathcal{B}$ also schedulable. If the packed pattern is schedulable and thanks to the previous lemma, then $\widehat{\widehat{\DBF}}(\tau_i, t) \leq \widehat{\widehat{\DBF}}(\tau_i^{p}, t) \leq t, \forall t$, that is, the pattern $\Sigma$ is schedulable, and thus, the packed pattern is worse than $\Sigma$. Since $\Sigma$ represent any pattern, the packed pattern is the worst-case pattern. The theorem follows.
\fin
\end{proof}

\end{appendix}
\end{document}